\documentclass[11pt, letterpaper,reqno]{amsart}

\usepackage{graphicx}
\usepackage[hyphens]{url}
\usepackage{appendix}
\usepackage{natbib}

\numberwithin{equation}{section}

\DeclareMathOperator*{\argmax}{arg\,max}

\newtheorem{theorem}{Theorem}[section]
\newtheorem{assumption}{Assumption}[section]
\newtheorem{corollary}{Corollary}[section]
\newtheorem{lemma}{Lemma}[section]
\newtheorem{proposition}{Proposition}[section]
\theoremstyle{definition}
\newtheorem{definition}{Definition}[section]
\theoremstyle{remark}
\newtheorem{remark}{Remark}[section]
\newtheorem{example}{Example}[section]

\newcommand{\leref}{Lemma~\ref}
\newcommand{\deref}{Definition~\ref}

\newcommand{\exref}{Example~\ref}
\newcommand{\prref}{Proposition~\ref}
\newcommand{\thref}{Theorem~\ref}

\newcommand{\asref}{Assumption~\ref}

\renewcommand{\P}{\mathbb{P}}
\newcommand{\Q}{Q}
\newcommand{\R}{\mathbb{R}}

\newcommand{\N}{\mathbb{N}}

\newcommand{\cM}{\mathcal{M}}

\newcommand{\kP}{\mathfrak{P}}
\newcommand{\cP}{\mathcal{P}}

\newcommand{\cI}{\mathcal{I}}
\newcommand{\cJ}{\mathcal{J}}
\newcommand{\eps}{\varepsilon}

\title[]{Super-hedging American Options with Semi-static Trading Strategies under Model Uncertainty}
\author[]{Erhan Bayraktar}\thanks{E. Bayraktar is supported in part by the National Science Foundation under grant DMS-1613170 and the Susan M. Smith chair.}
\address{Department of Mathematics, University of Michigan}
\email{erhan@umich.edu}
\author[]{Zhou Zhou}
\address{Institute for Mathematics and its Applications, University of Minnesota}
\email{zhouzhou@ima.umn.edu}
\begin{document}
\maketitle

\begin{abstract}
We consider the super-hedging price of an American option in a discrete-time market in which stocks are available for dynamic trading and European options are available for static trading. We show that the super-hedging price $\pi$ is given by the supremum over the prices of the American option under randomized models. That is, $\pi=\sup_{(c_i,Q_i)_i}\sum_ic_i\phi^{Q_i}$, where $c_i\in\R_+$ and the martingale measure $Q^i$ are chosen such that $\sum_i c_i=1$ and $\sum_i c_iQ_i$ prices the European options correctly, and $\phi^{Q_i}$ is the price of the American option under the model $Q_i$. Our result generalizes the example given in \cite{2016arXiv160402274H} that the highest model based price can be considered as a randomization over models. 
\end{abstract}

\section{introduction}

Recently, using very different techniques, \cite{ZZ4} and \cite{Neuberger} (finally published as \cite{Hobson3}) both calculated the super-hedging prices of American options when one is allowed to use semi-static trading strategies. In \cite{ZZ4}, the authors show that the super-hedging price (hedger's price) can be strictly greater than the highest model based price $\sup_Q\phi^Q$ (Nature's price), while in \cite{Hobson3}, the authors show that these two prices are equal. The reason of the difference lies in that, in \cite{ZZ4} the hedger and Nature have the same information/filtration, while in \cite{Hobson3} Nature has access to more information (i.e., Nature's filtration is larger than hedger's). As a result the definition of highest model price in the two papers are different (although the superhedging prices are the same). See \cite{2016arXiv160402274H} for another comparison of the two results. \cite{2016arXiv160402274H} also gives an example in which the highest model based price can be considered as a randomization over models. 

In this paper, we show that the super-hedging price $\pi$ is given by supremum over randomized model based prices of the American option. That is, $\pi=\sup_{(c_i,Q_i)_i}\sum_ic_i\phi^{Q_i}$, where $c_i\in\R_+$ and the martingale measure $Q^i$ are chosen such that $\sum_i c_i=1$ and $\sum_i c_iQ_i$ prices the European options correctly, and $\phi^{Q_i}$ is the price of the American option under the model $Q_i$.

Our result gives another representation of the super-hedging duality in \cite{ZZ4} and builds a connection with the main result of \cite{Hobson3}. Moreover, from our result we improve on the result in \cite{Hobson3}.  As indicated by our result, in terms of Nature's pricing, only the randomized models (see \deref{d2}) are relevant, and all the other models proposed in \cite{Hobson3} are redundant. Our result also generalizes the example provided in \cite{2016arXiv160402274H} that the highest model based price can be considered as a randomization over models.

After we wrote this note, Deng and Tan wrote a related paper, \cite{tan}, in which they showed the super-hedging price equals the supremum of the expectation of some related payoff over a suitable family of martingale measures. Their approach is to enlarge the probability space by the exercise time of the American option. The main result in our paper in contrast does not require the enlargement of the space. Let us also mention that even though both paper uses the terminology ``randomization'', it has totally different meanings. In \cite{tan}, randomization refers to the enlarged space, while in our paper it refers to mixing the martingale measures in the original space. The related problem of hedging path dependent options using American style options is considered in \cite{bayraktar2016, 2016arXiv160501327B}.

The paper is organized as follows. In the next section, we provide the setup and the main result. In Section 3, we give a discussion for the results in \cite{ZZ4}, \cite{Hobson3}, and this paper. Finally, we provide the proof for the main result in Section 4.

\section{Setup and main result}

We use the set-up in \cite{Nutz2} and \cite{ZZ4}. Let $T\in\mathbb{N}$ be the time horizon and let $\Omega_1$ be a Polish space. For $t\in\{0,1,\dotso,T\}$, let $\Omega_t:=\Omega_1^t$ be the $t$-fold Cartesian product, with the convention that $\Omega_0$ is a singleton. We denote by $\mathcal{F}_t$ the universal completion of $\mathcal{B}(\Omega_t)$ and write $(\Omega,\mathcal{F})$ for $(\Omega_T,\mathcal{F}_T)$. Denote $\mathbb{F}:=(\mathcal{F}_t)_{t=0,\dotso,T}$. Let $\mathfrak{P}(\Omega_1)$ be the set of all the probability measures on $(\Omega_1,\mathcal{B}(\Omega_1))$. For each $t\in\{0,\dotso,T-1\}$ and $\omega\in\Omega_t$, we are given a nonempty convex set $\mathcal{P}_t(\omega)\subset\mathfrak{P}(\Omega_1)$ of probability measures. We assume that for each $t$, the graph of $\mathcal{P}_t$ is analytic, which ensures that $\mathcal{P}_t$ admits a universally measurable selector, i.e., a universally measurable kernel $P_t:\ \Omega_t\rightarrow \mathfrak{P}(\Omega_t)$ such that $P_t(\omega)\in\mathcal{P}_t(\omega)$ for all $\omega\in\Omega_t$. Let
\begin{equation}\label{prob}
\mathcal{P}:=\{P_0\otimes\dotso\otimes P_{T-1}:\ P_t(\cdot)\in\mathcal{P}_t(\cdot),\ t=0,\dotso,T-1\},
\end{equation}
where each $P_t$ is a universally measurable selector of $\mathcal{P}_t$, and
$$P_0\otimes\dotso\otimes P_{T-1}(A)=\int_{\Omega_1}\dotso\int_{\Omega_1} 1_A(\omega_1,\dotso,\omega_T)P_{T-1}(\omega_1,\dotso,\omega_{T-1};d\omega_T)\dotso P_0(d\omega_1),\ \ \ A\in\Omega.$$

Let $S_t:\Omega_t\rightarrow\mathbb{R}^d$ be Borel measurable, which represents the price at time $t$ of a $d$-dimensional stock $S$ that can be traded dynamically in the market.  Let $g=(g_1,\dotso,g_e):\Omega\rightarrow\mathbb{R}^e$ be Borel measurable, representing the European options that can only be traded at the beginning at price $0$. Assume NA$(\mathcal{P})$ holds, i.e, for all $(H,h)\in\mathcal{H}\times\mathbb{R}^e$,
$$(H\cdot S)_T+hg\geq 0\ \ \ \mathcal{P}-\text{q.s.\ \ \ \ \ \ implies \ \ \ \ \ \ } (H\cdot S)_T+hg=0\ \ \ \mathcal{P}-\text{q.s.},$$
where $\mathcal{H}$ is the set of $\mathbb{F}$-predictable processes, $(H\cdot S)_T:=\sum_{t=0}^{T-1}H_t(S_{t+1}-S_t)$, and $h g$ denotes the inner product of $h$ and $g$. \footnote{We say that a set is $\mathcal{P}$ polar if it is $P$-null for all $P \in \mathcal{P}$. A property is said to hold $\mathcal{P}$- quasi-surely (q.s.) if it holds outside a $\mathcal{P}$-polar set.}
Then from \cite[FTAP]{Nutz2}, for all $P\in\mathcal{P}$, there exists $Q\in\mathcal{Q}$ such that $P\ll Q$, where 
$$\mathcal{Q}:=\{Q \text{ martingale measure}\footnote{That is, $Q$ satisfies $E_Q[|S_{t+1}|\ |\mathcal{F}_t]<\infty$ and $E_Q[S_{t+1}|\mathcal{F}_t]=S_t, \ Q$-a.s. for $t=0,\dotso,T-1$.}:\ E_Q[g]=0,\text{ and }\exists P'\in\mathcal{P}, \text{ s.t. } Q\ll P'\}.$$ For $t=0,\dotso,T$ and $\omega\in\Omega_t$, define 
$$\mathcal{Q}_t(\omega):=\{Q\in\mathfrak{P}(\Omega_1):\ Q\ll P, \text{ for some } P\in\mathcal{P}_t(\omega),\text{ and } E_Q[S_{t+1}(\omega,\cdot)-S_t(\omega)]=0\}.$$
By \cite[Lemma 4.8]{Nutz2}, there exists a universally measurable selector $Q_t$ such that $Q_t(\cdot)\in\mathcal{Q}_t(\cdot)$ on $\{\mathcal{Q}_t\neq\emptyset\}$. Using these selectors we define for $t\in\{0,\dotso,T-1\}$ and $\omega\in\Omega_t$,
$$\mathcal{M}_t(\omega):=\left\{Q_t\otimes\dotso\otimes Q_{T-1}:\ Q_i(\omega,\cdot)\in\mathcal{Q}_i(\omega,\cdot)\text{ on }\{\mathcal{Q}_i(\omega,\cdot)\neq\emptyset\},\ i=t,\dotso,T-1\right\},$$
which is similar to \eqref{prob} but starting from time $t$ instead of time 0. In particular $\mathcal{M}_0=\mathcal{M}$, where 
\begin{equation}\notag
\mathcal{M}:=\{Q \text{ martingale measure}: \exists P\in\mathcal{P}, \text{ s.t. }Q\ll P\}.
\end{equation} 
We assume that the graph of $\mathcal{M}_t$ is analytic, $t=0,\dotso,T-1$. A general sufficient condition for the analyticity of graph($\mathcal{M}_t)$ is provided in \cite[Proposition 1.1]{ZZ4}.

Let $\mathcal{T}$ be the set of $\mathbb{F}$-stopping times, and $\mathcal{T}_t$ be the set of $\mathbb{F}$-stopping times that are no less than $t$. Denote $||\cdot||$ for the Euclidean norm.

Let us consider an American option with pay-off stream $\Phi$. We will assume that $\Phi:\{0,\dotso,T\}\times\Omega\rightarrow\mathbb{R}$ is $\mathbb{F}$-adapted. We define the super-hedging price as
\begin{equation}\label{e5}
\pi(\Phi):=\inf\left\{x\in\mathbb{R}:\ \exists (\tilde H,h)\in\mathcal{H}'\times\mathbb{R}^e,\ \text{ s.t. } x+(\tilde H(t)\cdot S)_T+hg\geq \Phi_t,\ \mathcal{P}-q.s.,\ t=0,\dotso,T\right\},
\end{equation}
where
$$\mathcal{H}':=\{\tilde H=(H,H(0),\dotso,H(T))\subset\mathcal{H}^{T+2}\},$$
and
$$\tilde H_s(t):=H_s 1_{\{s<t\}}+H_s(t) 1_{\{s\geq t\}},\quad s=0,\dotso,T.$$
Here for $\tilde H=(H,H(0),\dotso,H(T))\in\mathcal{H}'$, $H$ represents the strategy that the hedger use before the American option is exercised, and $H(t)$ is the strategy the hedger will use after the American option is exercised at time $t$.

We make the following standing assumption throughout the paper.
\begin{assumption}\label{a1}
{\ }
\begin{itemize}
\item[(1)] For $t\in\{1,\dotso,T\}$ and $(\omega,P)\in\Omega_T\times\mathfrak{P}(\Omega_{T-t})$, the map $(\omega,P)\mapsto\sup_{\tau\in\mathcal{T}_t}E_P[\Phi_\tau(\omega^t,\cdot)]$ is upper-semianalytic, where $\omega^t$ is the path of $\omega$ up to time $t$.
\item[(2)] $\sup_{Q\in\mathcal{M}}E_Q[||g||]<\infty$ and $\sup_{Q\in\mathcal{M}}E_Q[\max_{0\leq t\leq T}|\Phi_t|]<\infty$.
\end{itemize}
\end{assumption}

\begin{remark}
If $\Phi_t$ is lower-semicontinuous and bounded from below for $t=0,\dotso,T$, then \asref{a1}(1) is satisfied. See \cite[Proposition 3.1]{ZZ4}. 
\end{remark}

Below is the super-hedging result from \cite[Theorem 3.1]{ZZ4}.
\begin{lemma}\label{l1}
Let \asref{a1} hold. Then
\begin{equation}\label{e1}
\pi(\Phi)=\inf_{h\in\mathbb{R}^e}\sup_{Q\in\mathcal{M}}\sup_{\tau\in\mathcal{T}}E_Q[\Phi_\tau-hg],
\end{equation}
Moreover, there exists $(H^*,h^*)\in\mathcal{H}'\times\mathbb{R}^e$, such that
\begin{equation}\notag
\pi(\Phi)+(H^*\cdot S)_T+h^*g\geq\Phi_\tau,\ \mathcal{P}-q.s.,\ \forall\tau\in\mathcal{T}.
\end{equation}
\end{lemma}

In this paper, we will get another representation for \eqref{e1} as the super-hedging duality. Below is our main result.

\begin{theorem}\label{t1}
Let \asref{a1} hold. Then
\begin{equation}\label{e2}
\pi(\Phi)=\sup_{(c_i,Q_i)_i}\sum_i c_i\sup_{\tau\in\mathcal{T}}E_{Q_i}[\Phi_\tau],
\end{equation}
where the supremum is over all finite sequence $(c_i,Q_i)_i$, such that for each $i$, $c_i\in(0,\infty)$, $Q_i\in\mathcal{M}$, $\sum_i c_i=1$ and $\sum_i c_i Q_i\in\mathcal{Q}$.
\end{theorem}
\begin{remark}
Let us call the dual of the super-hedging price ``the price given by Nature''. We can interpret \eqref{e2} as follows. Nature randomizes the models $(Q_i)_i \subset\mathcal{M}$ (not $\subset\mathcal{Q}$) in such a way that it appears to the hedger that the European options $g$ are priced correctly. Moreover, the true state $i$ is revealed at the beginning only to Nature, not to the hedger.
\end{remark}

\section{Comparison between results in \cite{ZZ4}, \cite{2016arXiv160402274H, Hobson3}, and \thref{t1}}

In \cite{ZZ4, 2016arXiv160402274H}, it is pointed out that the super-pricing price (hedger's price) can be strictly greater than the highest model based price (Nature's price). That is, it is possible that
\begin{equation}\label{eq:duality-gap}
\pi(\Phi)>\sup_{Q\in\mathcal{Q}}\sup_{\tau\in\mathcal{T}}E_Q[\Phi_\tau].
\end{equation}
That is it appears that there is a duality gap between the super-hedging price and the highest model price.

In \eqref{eq:duality-gap}, the hedger and Nature have the same power in the sense that both use the filtration (information) $\mathbb{F}$. So in order to make Nature's price (the highest model price) equal to the hedger's price, an intuitive way is to enhance the power of Nature by providing more information to Nature. \thref{t1} indicates that in addition to the information $\mathbb{F}$, if Nature can also use the information of the initial distribution of the possible models $(Q_i)_i$, then the hedger's price and Nature's price will be the same.

In contrast \cite{Hobson3}, Nature can have all sorts of information as long as the models appears to be consistent to the hedger (note that the hedger only knows the information generated by the stock). To be more precise, the authors of \cite{Hobson3} call a filtered probability space $M=(\Omega',\mathbb{F}'=(\mathcal{F}_t')_{t=0,\dotso,T},Q')$ a consistent model if the space $(\Omega',\mathbb{F}')$ supports a stochastic process $S$ (and random vector $g$), $S$ is a $(Q',\mathbb{F}')$-martingale, and $E_{Q'}[g]=0$. It is then shown in \cite{Hobson3}\footnote{ \cite{2016arXiv160402274H} demonstrates this on an example. It first observes \eqref{eq:duality-gap} and then shows that when the models are enlarged to the consistent models there is no duality gap. It seems the authors of \cite{Hobson3} were initially not aware of \cite{ZZ4} and developed their result independently with using very different techniques. Then they wrote \cite{2016arXiv160402274H} to clarify what initially may look like a contradiction since they had not observed a duality gap.} that the hedger's price based on the filtration generated by $S$ is equal to Nature's price (the highest model price) based on all consistent models (i.e., $\sup_M\phi^M(\Phi)$, where $\phi^M(\Phi)$ is the price of the American option based on the model $M$). Due to their proof approach, they only work on the canonical filtration for the hedger, and the European options are specified to be calls (or puts, or some other equivalent forms). Let us point out that \cite{Hobson3} has a second result, which says that the search over all probability spaces can be reduced to a search over a much simpler class of models (see \cite[Section 2.4]{Hobson3}).
 
In this paper, we show that starting from \cite{ZZ4} and using a (not so conventional) min-max, we can obtain the results of \cite{Hobson3} with a simpler proof and in fact obtain a stronger result: We demonstrate below (as an application of our main result) that only a very small subclasses of extensions which we call randomized models are relevant in representing the superhedging price. We need some preparation to describe our generalization. First, as a counterpart of \cite[Definition 1]{Hobson3}, let us first provide the following definition. Below by ``embedding'' we mean that $\Omega$ can be regarded as a sub space of $\Omega'$ and $\mathcal{F}_t$ can be regarded as a sub sigma algebra of $\mathcal{F}_t'$. (See \deref{d2} for an example.)

\begin{definition}\label{d1}
We say that a filtered probability space  $(\Omega',\mathbb{F}'=(\mathcal{F}_t')_{T=0,\dotso,T},Q')$ is belongs to the set of Nature's models if
\begin{itemize}
\item[(1)] The filtered space $(\Omega,\mathbb{F})$ can be embedded in $(\Omega',\mathbb{F}')$.
\item[(2)] (After embedding) $S$ is a $(Q',\mathbb{F}')$-martingale, and $E_{Q'}[g]=0$.
\item[(3)] (After embedding) for any $A\in\mathcal{F}_T$, if $\sup_{P\in\cP}P(A)=0$ then $Q'(A)=0$.
\end{itemize}
Denote $\mathbb{M}_n$ to be the collection of all such models.
\end{definition}

\begin{remark}
Here the hedger's filtration is $\mathbb{F}$, which is not necessarily the canonical filtration generated by $S$. By \deref{d1}(1), Nature can have a strictly larger filtration $\mathbb{F}'$.
\end{remark}

\begin{definition}\label{d2}
We call $(\Omega',\mathbb{F}'=(\mathcal{F}_t)_{T=0,\dotso,T},Q')$ a randomized model, if there exist $n\in\mathbb{N}$, $(c_i)_{i=1,\dotso,n}\subset(0,\infty)$ with $\sum_{i=1}^nc_i=1$, $(Q_i)_{i=1,\dotso,n}\subset\mathcal{M}$ with $\sum_{i=1}^nc_iE_{Q_i}[g]=0$, such that
\begin{itemize}
\item[(1)] $\Omega'=\Omega\times\{1,\dotso,n\}$,
\item[(2)] $\mathcal{F}_t'=\mathcal{F}_t\otimes\mathcal{B}(\{0,\dotso,n\})$.
\item[(3)] For $i\in\{1,\dotso,n\}$, the transaction kernel $Q'(\cdot,i)=Q_i(\cdot)$, and $Q'(\{i\})=c_i$.
\end{itemize}
Denote $\mathbb{M}_r$ as the set of all randomized models.
\end{definition}
\begin{remark}
It is easy to see that $\mathbb{M}_r\subset\mathbb{M}_n$.
\end{remark}

\begin{corollary}\label{c1}
The super-hedging price defined by \eqref{e5} is given by
$$\pi(\Phi)=\sup_{M\in\mathbb{M}_n}\phi^M(\Phi)=\sup_{M\in\mathbb{M}_r}\phi^M(\Phi),$$
where $\phi^M(\Phi)$ is the price of the American option under model $M$. That is, for $M=(\Omega',\mathcal{F}',Q')\in\mathbb{M}_n$,
$$\phi^M(\Phi):=\sup_{\tau\ \mathbb{F}'\text{-stopping time}}E_{Q'}[\Phi_\tau].$$
\end{corollary}
\begin{proof}
It is easy to show that
$$\pi(\Phi)\geq\sup_{M\in\mathbb{M}_n}\phi^M(\Phi)\geq\sup_{M\in\mathbb{M}_r}\phi^M(\Phi),$$
and by \thref{t1} we have that
$$\pi(\Phi)=\sup_{M\in\mathbb{M}_r}\phi^M(\Phi).$$
\end{proof}


\section{Proof of \thref{t1}}
Without loss of generality, we assume that any of these European options cannot be replicated by stock $S$ and other European options. That is, for any $(H,h)\in\mathcal{H}\times\R^e$,
\begin{equation}\label{e4}
\text{if}\ \ (H\cdot S)_T+hg=0\ \ \mathcal{P}-\text{q.s.},\quad\quad\text{then}\ \ h=0.
\end{equation}
For otherwise, we can work on the new market where the redundant European options are removed, and obviously the super-hedging price $\pi$ would be the same. Moreover, the set of the sequences $(c_i,Q_i)_i$ in the supremum in \eqref{e2} would be unchanged.

We will prove that
$$\inf_{h\in\mathbb{R}^e}\sup_{Q\in\mathcal{M}}\sup_{\tau\in\mathcal{T}}E_Q[\Phi_\tau-hg]=\sup_{(c_i,Q_i)}\sum_i c_i\sup_{\tau\in\mathcal{T}}E_{Q_i}[\Phi_\tau].$$
To this end, denote $\mathcal{M}=(Q_\alpha)_{\alpha\in I}$. Let
$$A:=\{(c_\alpha)_{\alpha\in I}:\ c_\alpha\in\R,\ \text{only finitely many }c_\alpha\neq 0\}.$$
Let $d:A\times A\mapsto\R$,
$$d(c,c'):=\sum_{\alpha\in I}|c_\alpha-c_\alpha'|,\quad c=(c_\alpha)_{\alpha\in I},c'=(c_\alpha')_{\alpha\in I}\in A.$$
Then it is easy to see that $d$ defines a metric, and thus $(A,d)$ is Hausdorff topological vector space. Let
$$X:=\left\{(c_\alpha)_{\alpha\in I}\in A:\ c_\alpha\geq 0,\ \sum_\alpha c_\alpha=1\right\}.$$

Obviously, we have that
\begin{equation}\label{e3}
\inf_{h\in\mathbb{R}^e}\sup_{Q\in\mathcal{M}}\sup_{\tau\in\mathcal{T}}E_Q[\Phi_\tau-hg]=\inf_{h\in\mathbb{R}^e}\sup_{c\in X}\left[\sum_\alpha c_\alpha\sup_{\tau\in\mathcal{T}}E_{Q_\alpha}[\Phi_\tau-hg]\right].
\end{equation}
Now, if the inf and sup at the right hand side of \eqref{e3} can be exchanged without changing the value, then we have that
$$\inf_{h\in\mathbb{R}^e}\sup_{Q\in\mathcal{M}}\sup_{\tau\in\mathcal{T}}E_Q[\Phi_\tau-hg]=\sup_{c\in X}\inf_{h\in\mathbb{R}^e}\left[\sum_\alpha c_\alpha\sup_{\tau\in\mathcal{T}}E_{Q_\alpha}[\Phi_\tau-hg]\right]=\sup_{(c_i,Q_i)}\sum_i c_i\sup_{\tau\in\mathcal{T}}E_{Q_i}[\Phi_\tau],$$
where for the second equality, we use the fact that if $\sum_\alpha E_{Q_\alpha}[g]\neq 0$, then we can push the value inside the sup\,inf to $-\infty$ by choosing $h$ properly.

In the rest of the proof, we will show that the sup and inf at the right hand side of \eqref{e3} can be exchanged. Since we do not have the compactness of the underlying sets $R^e$ and $X$, we will apply minimax theorem \cite[Theorem 2]{{MR638178}} (we provide it the appendix as \thref{t2}.).

Let $f:X\times\R^e\mapsto\R$,
$$f(c,h):=\sum_\alpha c_\alpha\sup_{\tau\in\mathcal{T}}E_{Q_\alpha}[\Phi_\tau-hg].$$
It is easy to see that for $c\in X$ and $h\in\R^e$, the maps $c\mapsto f(c,h)$ and $h\mapsto f(c,h)$ are linear. For $c,c'\in X$ and $h\in\R^e$,
$$f(c,h)-f(c',h)\leq\sum_\alpha|c_\alpha-c_\alpha'|\left|\sup_{\tau\in\mathcal{T}}E_{Q_\alpha}[\Phi_\tau]-h E_{Q_\alpha}[g]\right|\leq C_h\sum_\alpha|c_\alpha-c_\alpha'|,$$
where
$$C_h:=\sup_{Q\in\mathcal{M}}E_Q\left[\max_{0\leq t\leq T}|\Phi_t|\right]+||h||\sup_{Q\in\mathcal{M}}E_Q||g||<\infty$$
by \asref{a1}(2). Therefore, the map $c\mapsto f(c,h)$ is continuous. Similarly, for any $h,h'\in\R^e$ and $c\in X$,
$$|f(c,h)-f(c,h')|\leq\sum_\alpha c_\alpha E_{Q_\alpha}|(h-h')g|\leq ||h-h'||\sup_{Q\in\mathcal{M}}E_Q||g||,$$
which implies that the map $h\mapsto f(c,h)$ is continuous.

We claim that $0\in\R^e$ is an interior point of the convex set $\{E_Q[g]:\ Q\in\mathcal{M}\}$. If not, then there would exist a non-zero vector $h^*\in\R^e$, such that $h^* E_Q[g]\leq 0$ for any $Q\in\mathcal{M}$. Then by \cite[superhedging theorem]{Nutz2}, the super-hedging price of $h^*g$ using only stock $S$ would be no greater than $0$. Moreover, there would exist $H\in\mathcal{H}$ such that
$$(H\cdot S)_T\geq h^*g,\quad\cP\text{-q.s.}.$$
Then NA$(\cP)$ would imply that
$$(H\cdot S)_T-h^*g=0,\quad\cP\text{-q.s.}.$$
This contradicts \eqref{e4} since $h^*\neq 0$.

Let $\bar B(r)$ be a closed ball in $\R^e$ with radius $r$ centered at the origin, where $r>0$ is chosen such that $\bar B(r)\subset\{E_Q[g]:\ Q\in\mathcal{M}\}$. Denote $L:=1+\sup_{Q\in\mathcal{M}}E_Q\left[\max_{0\leq t\leq T}\Phi_t\right]<\infty$. Let $J:=\bar B(3L/r)$. Since for any $h\in\partial J$, there exist $Q^h\in\mathcal{M}$ such that $||E_{Q^h}[g]||=r$ and $h E_{Q^h}[g]=-3L$. 
Then for $h\in\partial J$, there exists an open ball $B(h,\eps^h)$ with radius $\eps^h>0$ centered at $h$, such that
$$h' E_{Q^h}[g]\leq -2L,\quad\forall h'\in B(h,\eps^h).$$
Since $\partial J\subset\cup_{h\in\partial J} B(h,\eps^h)$, there exists a finite set $(h_i)_{i=1,\dotso,n}\subset\partial J$ such that $\partial J\subset\cup_{i=1}^n B(h_i,\eps^{h_i})$. For $i=1,\dotso,n$, let $\alpha_i$ be the index of $Q^{h_i}$ in $\mathcal{M}=\{(Q_\alpha)_\alpha:\ \alpha\in I\}$. As a result the set
$$K:=\{(c_\alpha)_\alpha\in X:\ c_{\alpha}=0\text{ if }\alpha\notin \{\alpha_i :  i=1,\dotso,n\} \}$$
is compact. For any $h\in J^c:=\R^e\setminus J$ there exists some $i\in\{1,\dotso,n\}$ such that $\frac{3L/r}{||h||}h\in B(h_i,\eps^{h_i})$, which implies
$$h E_{Q^{h_i}}[g]\leq -2L\frac{||h||}{3L/r}\leq -2L.$$
Therefore,
$$\inf_{h\in J^c}\sup_{c\in K}f(c,h)\geq\inf_{h\in J^c}\sup_{c\in K}\sum_\alpha c_\alpha E_{Q_\alpha}[-hg]-L\geq L\geq\sup_{c\in X}f(c,0)\geq\inf_{h\in\mathbb{R}^e}\sup_{c\in X}f(c,h).$$
Applying \thref{t2}, we have that
$$\inf_{h\in\mathbb{R}^e}\sup_{c\in X}f(c,h)=\sup_{c\in X}\inf_{h\in\mathbb{R}^e}f(c,h),$$
and this completes the proof.
\hfill $\square$

\begin{appendix}
\section{A non-compact minimax theorem}
Below is a minimax theorem without compactness of the underlying sets from \cite{MR638178}.
\begin{theorem}\label{t2}
Let $X, Y$ be nonempty convex sets, each in a Hausdorff topological vector space, and let $f$ be a real-valued function defined on $X\times Y$ such that

(a) For each $x\in X$, $f(x,y)$ is lower semi-continuous and quasi-convex on $Y$;

(b) For each $y\in Y$, $f(x,y)$ is upper semi-continuous and quasi-concave on $X$.\\
If there exists a nonempty compact convex set $K$ in $X$ and a compact set $H$ in $Y$ such that
$$\inf_{y\in Y}\sup_{x\in X}f(x,y)\leq\inf_{y\notin H}\max_{x\in K}f(x,y),$$
then
$$\inf_{y\in Y}\sup_{x\in X}f(x,y)=\sup_{x\in X}\inf_{y\in Y}f(x,y).$$
\end{theorem} 
\end{appendix}

\bibliographystyle{agsm}
\bibliography{ref}

\end{document}